\documentclass[journal]{IEEEtran}
\usepackage{amssymb}
\usepackage{amsmath}
\usepackage{amsfonts}
\usepackage{graphicx}
\usepackage{algorithm}
\usepackage{algorithmic}
\usepackage{epstopdf}
\usepackage{cite}
\usepackage{exscale}
\usepackage{relsize}
\usepackage{caption}

\usepackage{color}
\usepackage{multirow}
\usepackage{subfigure}
\usepackage{bm}
\usepackage{diagbox}
\newtheorem{proposition}{\textbf{Proposition}}

\begin{document}
\title{On Optimizing Hierarchical Modulation in AWGN channel}
\author{Baicen Xiao,  Kexin Xiao, Zhiyong Chen, Hui Liu\\
Department of Electronic Engineering, Shanghai Jiao Tong University, Shanghai, P. R. China\\
Email: {\{xinzhiniepan, kexin.xiao, zhiyongchen, huiliu\}@sjtu.edu.cn}}
\maketitle
\begin{abstract}
Hierarchical modulation (HM) is able to provide different levels of protection for data streams and achieve a rate region that cannot be realized by traditional orthogonal schemes, such as time division (TD). Nevertheless, criterions and algorithms for general HM design are not available in existing literatures. In this paper, we jointly optimize the constellation positions and binary labels for HM to be used in additive white gaussian noise (AWGN) channel. Based on bit-interleaved coded modulation (BICM) with successive interference cancellation (SIC) capacity, our main purpose is to maximize the rate of one data stream, with power constrains and the constrain that the rate of other data streams should be larger than given thresholds. Multi-start interior-point algorithm is used to carry out the constellation optimization problems and methods to reduce optimization complexity are also proposed in this paper. Numerical results verify the performance gains of optimized HM compared with optimized quadrature amplidude modulation (QAM) based HM and other orthogonal transmission methods.
\end{abstract}
\section{INTRODUCTION}
In order to meet the requirement of high spectrum efficiency in 5G wireless communication system and the next generation of broadcasting system, there are tremendous interests in hierarchical modulation (HM). HM is a promising modulation scheme  to carry multiple data streams efficiently through one hierarchical symbol stream, which is significantly useful for many multi-data systems, such as broadcast systems \cite{upgrade,MBMS} and some other multi-receiver systems \cite{Relay,power-allocation}. From the perspective of information theory, HM can be viewed as a kind of superposition coding, which can improve the throughput of the whole system compared with time division multiplexing method, especially when the receiving signal-to-noise ratio (SNR) for every user is quite different \cite{broadcast,IT,wireless}. Moreover, HM is frequently introduced to provide unequal protection by taking advantage of the property that different bit positions of the labels for HM (different positions can be considered as different sub-channels) are able to possess different reliabilities, i.e. different error probabilities \cite{Hier-ID, Image}.

Characterized by many attractive properties, HM has been extensively studied. HM has been adopted by several standards, such as DVB-SH/S2 \cite{ETSI1,ETSI2} in order to provide different level of protection. In \cite{upgrade}, Jiang points out that HM can provide backward compatibility, which is a difficulty in increasing the bit rate of a fixed rate digital broadcast system, so that both old receivers and new receivers can work properly. \cite{voice} exploits the layered property of HM to transmit multimedia data according to channel conditions in multiclass data transmission where there are usually several kinds of data with different priorities. In \cite{Relay}, the authors describe how to use HM at the source in relay communication to achieve performance gain. Although HM has different data layers, some data layers are invalid if the corresponding sub-channel condition is bad, and hence the throughput is low. To tackle this problem and to provide higher throughput of the whole system, \cite{Adaptive} combines time sharing and HM and shows that a performance gain of roughly $15\%$ can be achieved by grouping users properly. In order to evaluate the bit error rate (BER), \cite{ber} and \cite{Hier-ID} carry out the approximate BER calculation method for uncoded HM and bit interleaved coded modulation (BICM) with iterative decoding (ID) HM respectively. Especially, \cite{Hier-ID} investigates different labelling rules for HM to identify which can offer good performance when ID is adopted.

Even though there are a large number of works related to HM, most of these works only consider uniform HM, e.g. HM based on adjusting uniform quadrature amplidude modulation (QAM). And to optimize HM, existing works mainly focus on adjusting labelling method based on H-QAM like \cite{Hier-ID}. However, QAM cannot take full advantage of shaping gain which can be exploited to improve spectrum efficiency \cite{shaping,Non-uniform}. Moreover, QAM usually has a large peak-to-average power ratio (PAPR) which may deteriorate performance when nonlinear power amplifier is adopted. And even for QAM-based HM (hereafter this is called H-QAM), there lacks satisfying criterions and algorithms to optimize layer priority factor. In \cite{power-allocation} Choi points out the problem of optimizing priority factor (also called power allocation factor) of H-QAM is nontrivial but does not carry out a more general algorithm for the optimization. Although the constellation set for HM proposed by Meric et al. is based on APSK and can improve throughput, the constellation set is still based on fixed number of rings which will limit the freedom, and no method is provided in their work to design general constellation set for HM. Some scholars \cite{Non-uniform, phase} deal with the constellation optimization problem for non-hierarchical modulation in different scenarios, where both labels and constellation positions are considered. However, no work aims to jointly optimize labels and signal positions for HM constellation and there is even no related design criterion.

In this paper, in order to jointly optimize the labels and positions of constellation for HM, we first derive the BICM with successive interference cancellation (SIC) capacity for HM in AWGN channel. Then based on BICM-SIC capacity, we formulate the goal (or say criterion) under which we carry out joint constellation label and position optimization. Furthermore, we propose to apply multi-start interior-point method to solve the non-convex HM optimization problem. To make the algorithm more efficient, the proposition that enlarging constellation for HM will not decrease the BICM-SIC capacity is derived. And to lower the complexity of optimization, simplification methods based on symmetric assumption are also proposed. To sum up, the main contributions of this paper are listed as follows:
\begin{itemize}
    \item Optimization criterion based on BICM-SIC capacity is proposed for designing HM to be used in AWGN channel.
    \item Multi-start algorithm based on interior-point method is investigated to solve HM optimization problem.
    \item Proposition that enlarging constellation for HM will not decrease the BICM-SIC capacity is presented to facilitate the choice of initial feasible solution, and simplification methods are proposed to reduce complexity of the optimization procedure.

\end{itemize}

\section{System model and Problem formulation}
For the sake of clarity, throughout this paper we use lowercase letters and uppercase letters to denote deterministic variables and random variables respectively. Normal and boldface letters are used to denote scalars and vectors respectively. And the term ``bit" is adopted to denote a binary digit and the term ``information bit" is used to denote binary information unit. $\parallel\cdot\parallel$ denotes Euclidian (or $\ell_2$) norm.
\subsection{System model}

The structure of HM is shown in Fig. \ref{scheme2}. The number of bit streams is $k$ and the $i$-th bit streams are denoted as $\mathbf{b}_{i}$. The $i$-th bit stream $\mathbf{b}_{i}$ is encoded by the encoder $\text{ENC}_{i}$ into bit stream $\mathbf{c}_{i}$. Then the bit cells from $k$ streams are combined into a new bit cell with $m={\Sigma_{i=1}^k m_i}$ bits and then the new bit cell is mapped by $\mu$ into an $n$-dimensional hierarchical symbol $\mathbf{z} \in \mathbb{X} \subseteq\mathbb{C}^n$. Note that different bit streams can be protected unequally if the mapping method is designed properly, and hence the coding rates of different bit streams should be chosen accordingly.

\begin{figure}[t]
\centering
\includegraphics[width=3.5in]{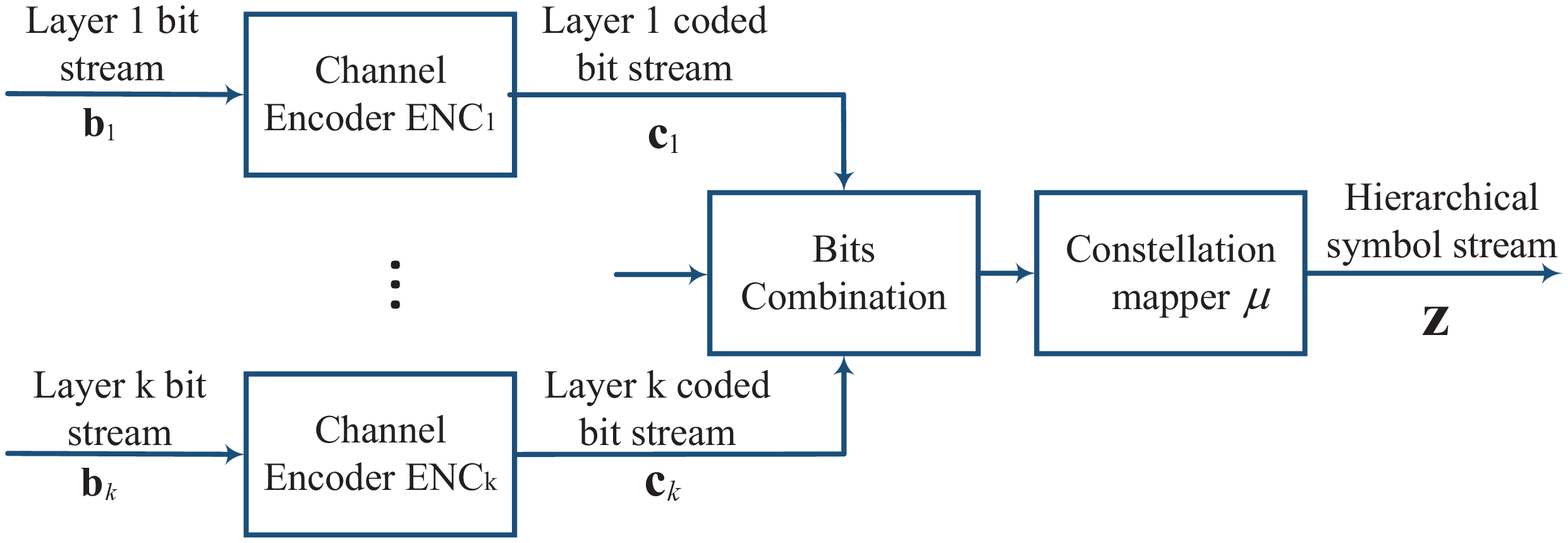}
\caption{Structure for hierarchical modulation.}
\label{scheme2}
\end{figure}

The hierarchical symbols $\mathbf{z}$ are transmitted through a vector channel characterized by transition probability density function (pdf.)
\begin{equation}
p(\mathbf{y}|\mathbf{z}):\mathbf{z},\mathbf{y}\in\mathbb{C}^n.
\end{equation}
If $n>1$, this channel model can represent MIMO channels, however in this paper we assume $n=1$ for the sake of clarity and it is easy to expand those methods in this paper to the cases where $n>1$.

\subsection{Problem formulation}

In this paper, we consider the scenario where there are two data streams to be sent (hence $k=2$ and these two bit streams are named high priority (HP) and low priority (LP) streams respectively) and two representative signal-to-noise ratios (SNR) of user side are considered. As shown in Fig. \ref{structure}, we assume $snr_L$ is larger than $snr_H$. Those users whose $snr$ is higher than $snr_L$ are supposed to receive both HP stream and LP stream and are denoted as L-user. Those users whose $snr$ is between $snr_L$ and $snr_H$ are supposed to receive only HP stream and denoted as H-user. It should be noted that the problem formulation and solution methods in this paper can be expanded to the cases where $k>2$.

\begin{figure}[t]
\centering
\includegraphics[width=2.2in]{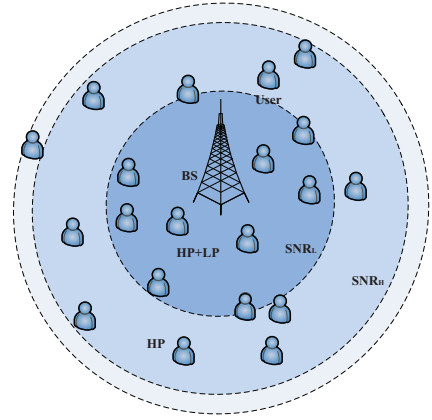}
\caption{The system structure considered in this paper.}
\label{structure}
\end{figure}

Throughout this paper, we consider additional white Gaussian noise (AWGN) channel, and the signal received by users with better $snr_H$ can be written as:
\begin{equation}
Y_L = Z+N_L, ~~~N_L\sim\mathcal{N}_c(0,\sigma_{L}^{2})
\end{equation}
and the signal received by users with worse $snr_H$ can be written as:
\begin{equation}
Y_H = Z+N_H, ~~~N_H\sim\mathcal{N}_c(0,\sigma_{H}^{2}).
\end{equation}
Here, $\sigma_{H}^{2}\geq\sigma_{L}^{2}$ and this channel can be viewed as degraded broadcasting channel \cite{IT}. And the transition pdf. of the AWGN channel can be written as
\begin{equation}
p(y|z)=\frac{1}{2\pi\sigma^2}e^{-\frac{\parallel y-z\parallel^2}{2\sigma^2}}
\end{equation}
Also, the methods in this paper can be applied to some fading channels, such as Rayleigh channel.

In some cases, better performance can be achieved by coded modulation (CM) \cite{trellis}. However, in most practical scenarios we use bit interleaved coded modulation (BICM) which is much simpler for real application and hence we take BICM capacity \cite{BICM} into consideration. Let $b_H^i$ denote the $i$-th bit of the binary label associated to HP stream and $b_L^j$ denote the $i$-th bit of the binary label associated to LP stream. For example, if $m_H=3$ and $m_L=2$, then we can use a five bits label $(b^1b^2b^3b^4b^5)$ to represent a hierarchical symbol. Without loss of generality, we assume the first 3 bits and last 2 bits represent HP stream and LP stream respectively, then a hierarchical symbol can be represented as $(b_H^1b_H^2b_H^3b_L^1b_L^2)$ or $(\mathbf{b}_H\mathbf{b}_L)$. The information bits which can be carried in HP stream can be calculated as,
\begin{align}\label{BICM_H}
r_{H}&=\sum\limits_{i=1}^{m_H}I(B_H^i;Y_H)
\end{align}
The information which can be carried in LP stream can be calculated as,
\begin{align}\label{BICM_L}
r_{L}&=\sum\limits_{j=1}^{m_L}I(B_L^j;Y_L|\mathbf{B}_H)
\end{align}
It is a reasonable assumption that every bit has a uniform input probability, then
\begin{align}
I(B_H^i;Y_H)=\!1+\!\frac{1}{2}\! \int_{-\infty}^{+\infty}\sum\limits_{b_H^i}p(y_H|b_H^i)\mathrm{log}_{2}\frac{p(y_H|b_H^i)}{\sum\limits_{\hat{b}_H^i}p(y_H|\hat{b}_H^i)}\mathrm{d}y_H
\end{align}
and $I(B_L^j;Y_L|\mathbf{B}_H)$ is given in equation (\ref{I_L}) where $\mathbb{X}_{\mathbf{b}_H,b_L^j}$ denotes the subset of hierarchical constellation symbols whose label has the value $b_L$ in position $j$ and the value $\mathbf{b}_H$ in positions corresponding to HP stream.
\begin{figure*}
\begin{align}\label{I_L}
I(B_L^j;Y_L|\mathbf{B}_H)=\sum\limits_{\mathbf{b}_H}p(\mathbf{b}_H)&I(B_L^j;Y_L|\mathbf{b}_H)= 1+ \frac{1}{2^{m_H}}\sum\limits_{\mathbf{b}_H} \int_{-\infty}^{+\infty}\sum\limits_{b_L^j}p(b_L^j,y_L|\mathbf{b}_H)\mathrm{log}_{2}p(b_L^j|y_L,\mathbf{b}_H)\mathrm{d}y_L \nonumber\\
&= 1+\frac{1}{2^{m_L+m_H}}\sum\limits_{\mathbf{b}_H}\int_{-\infty}^{+\infty}\sum\limits_{b_L^j}
\mathrm{log}_{2}\frac{\sum\limits_{z\in\mathbb{X}_{\mathbf{b}_H,b_L^j}}p(y_L|z)}{\sum\limits_{\hat{b}_L^j}\sum\limits_{z\in\mathbb{X}_{\mathbf{b}_H,\hat{b}_L^j}}p(y_L|z)}\sum\limits_{z\in\mathbb{X}_{\mathbf{b}_H,b_L^j}}p(y_L|z)\mathrm{d}y_L.
\end{align}
\hrulefill
\end{figure*}

According to actual requirements, under certain power constrains, we can formulate a criterion which maximizes the achievable rate that can be obtained by users who are able to receive both HP stream and LP stream, and guarantee the data rate $r_H$ of users with worse channel condition greater than a threshold $r^{*}$ in the meanwhile. That criterion can be formulated as the following optimization problem,
\begin{align}\label{max_p}
&\max_{z} ~~~~r_L\\
&s.t.~~\Bigg\{
\begin{array}{l}
r_H\geq r^{*};\\
\mathbb{E}[\parallel \!Z\!\parallel^2]\leq p.\nonumber
\end{array}
\end{align}
\begin{proposition}\label{ccc}
For optimization problem (\ref{max_p}), the constraint $\mathbb{E}[\parallel \!Z\!\parallel^2]\leq p$ is equivalent to $\mathbb{E}[\parallel \!Z\!\parallel^2]= p$ in AWGN channel.
\end{proposition}
\begin{proof}
See Appendix A.
\end{proof}

Proposition \ref{ccc} shows that inequality power constrain is equivalent to equality power constrain. However, inequality power constrain will make it more convenient for us to apply the optimization algorithm. To be specific, inequality power constrain will give us more freedom to choose an initial feasible solution. And this property makes it much easier for us to choose multiple initial feasible solutions to achieve a satisfied optimized constellation (the reason why we choose multiple initial solutions will be detailed in the next subsection).

Furthermore, for many actual situations, nonlinear power amplifiers have to be used, and hence peak-to-power ratio (PAPR) is an important factor that we should take into consideration. If the maximum PAPR is set to $\xi$, then the PAPR constrain can be defined as
\begin{align}\label{PAPR}
\mathrm{PAPR} = \frac{\max_{z\in\mathbb{X}}{\parallel \!z\!\parallel^2}}{\mathbb{E}[\parallel \!Z\!\parallel^2]}\leq\xi
\end{align}
It can be proved in a similar way that the proposition 1 also holds when PAPR constrain exists.

Because the optimized hierarchical constellations are very likely to be nonuniform and hence are costly for detection, we also optimize the scaling factor pair $(d_1, d_2)$ of H-QAM \cite{Hier-ID} in a similar way. Optimized H-QAM can be decomposed into two parts, in-phase and quadrature phase, greatly reducing the detection complexity. It should be emphasized that this kind of H-QAM is suitable for those HMs whose $m_H=2$.

\subsection{Algorithm for constellation optimization}
A $2^m$-order hierarchical constellation can be represented as $\mathbf{z}=(z_1,z_2,...,z_{2^m})^T\in\mathbb{C}^{2^m}$, and the problem comes down to optimizing the value and label of every element of $\mathbf{z}$ to maximize $r_L$ under some constrains. We can devide $2^m$ elements into $2^{m_H}$ groups with equal number of elements (i.e., each group has $2^{m_L}$ elements). When BICM capacity is considered, every symbol should be labeled. And without loss of generality, we label every element according to natural order. That is to say, we let binary vector $\mathbf{b}_H$ labels the $l(\mathbf{b}_H)$-th group and binary vector $\mathbf{b}_L$ labels the $l(\mathbf{b}_L)$-th element of each group, where $l(\mathbf{b})$ means the decimal representation of $\mathbf{b}$. For example, if $m_L=3$ and $m_H=2$, the grouping method is shown in Fig. \ref{constel_grouping2}. During the optimization process, the label of every symbol does not change and we just need to change the values of symbols to realize our goals. After the aforementioned arrangement, we can use (\ref{BICM_L}) to calculate BICM capacity for any given $\mathbf{z}$.

In order to apply optimization algorithm, $\mathbf{z}$ is reformulated as a real-valued vector $\mathbf{z}_r=(real(z_1),imag(z_1),...,real(z_m),imag(z_m))^T$.
\begin{figure}[t]
\centering
\includegraphics[width=3.0in]{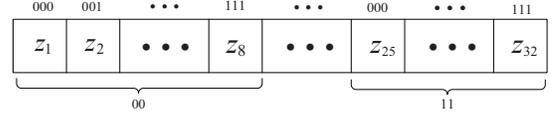}
\caption{Constellation grouping for joint label and position optimization, where $m_L=3$ and $m_H=2$. And $z_8$ is labeled by (00111) and $z_{25}$ is labeled by (11000).}
\label{constel_grouping2}
\end{figure}
For the sake of clear description of the algorithm we adopt, probelm (\ref{max_p}) is rewritten in the following form,
\begin{align}\label{mini_1}
&\min_{\mathbf{z}_r} ~~~~f_0(\mathbf{z}_r)\\
&s.t.~~~~~~f_i(\mathbf{z}_r)\leq 0,~~~i=1,\cdots,k. \nonumber
\end{align}
If there are no PAPR constrain, $k=2$, otherwise, $k=3$. That is a non-convex problem and can be solved by stochastic optimization algorithms, such as genetic algorithm and simulated annealing algorithm. However, stochastic algorithms usually cannot guarantee a global optimal solution and have high time complexity. In order to obtain satisfied results within an acceptable complexity, we use deterministic optimization algorithm combined with random multiple start points. For a deterministic optimization algorithm, every feasible start points is related to a local optimal solution, and hence the more different start points are chosen the more likely we can achieve a global optimal solution.

For every start point, we need to solve (\ref{mini_1}) to get a local optimal solution, which can be approximately achieved by solving a series of modified Karush-Kuhn-Tucker (KKT) condition. The KKT condition can be iteratively solved by primal-dual interior-point methods \cite{convex} and in each iteration the step ($\Delta\mathbf{z}_r,\Delta\boldsymbol{\lambda}$) can be got by solving equation (\ref{kkt_2}),
\begin{figure*}
\begin{align}\label{kkt_2}
\underbrace{\Bigg[
\begin{array}{l}
\nabla^2 f_0(\mathbf{z}_r)+\sum_{i=1}^m\lambda_i\nabla^2 f_i(\mathbf{z}_r)~~~D\boldsymbol{f}(\mathbf{z}_r)^T\\
-\text{diag}(\boldsymbol{\lambda})D\boldsymbol{f}(\mathbf{z}_r)~~~~~~~~~~~-\text{diag}(\boldsymbol{f}(\mathbf{z}_r))
\end{array}\Bigg]}_{\mathbf{m}}
\Bigg[
\begin{array}{l}
\Delta\mathbf{z}_r\\
\Delta\boldsymbol{\lambda}
\end{array}\Bigg]=
\Bigg[
\begin{array}{l}
\mathbf{r}_{pri}\\
\mathbf{r}_{dual}
\end{array}\Bigg].
\end{align}
\hrulefill
\end{figure*}
where $\boldsymbol{\lambda}=(\lambda_1,\cdots,\lambda_k)^T$, $\boldsymbol{f}(\mathbf{z}_r)=(f_1(\mathbf{z}_r),\cdots,f_k(\mathbf{z}_r))^T$, $\mathbf{r}_{pri} = \nabla f_0(\mathbf{z}_r)+\sum_{i=1}^k\lambda_i\nabla f_i(\mathbf{z}_r)$ and $\mathbf{r}_{dual}=-\text{diag}(\boldsymbol{\lambda})\boldsymbol{f}(\mathbf{z}_r)-(1/\mu)\mathbf{1}$. In some iterations, the leftmost matrix $\mathbf{m}$ in (\ref{kkt_2}) may be a singular matrix, and if this happens we can use pseudo inversion or turn to a more costly trust-region-based interior-point method \cite{trust_region} to update step.
\begin{figure}[t]
\centering
\includegraphics[width=1.6in]{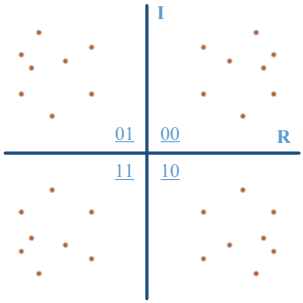}
\caption{Constellation under central symmetric assumption when $m_H=2$ and $m_L=3$. }
\label{symmetric}
\end{figure}

\textbf{Method to reduce complexity}: For large constellations, it is reasonable to assume that the whole constellation has some kind of symmetry property and we can take advantage of the symmetry to lower the number of variables we need to optimize. For example, if $m_H = 2$ and $m_L=3$, there are $4$ groups and we can assume that the groups are central symmetric as shown in Fig. \ref{symmetric}. We can reduce the number of complex variables from 32 to 8 and just optimize the value of those symbols which are in the first group whose $\mathbf{b}_H=00$. When the symbols in first cluster $\mathbf{z}_{c1}$ are determined, the symbols of the whole constellation is $(\mathbf{z}_{c1}^T, -\mathrm{conj}(\mathbf{z}_{c1})^T, \mathrm{conj}(\mathbf{z}_{c1})^T, -\mathbf{z}_{c1}^T)^T$ where $\mathrm{conj}(\cdot)$ is conjugate operation. Under symmetry assumption, we can greatly reduce optimization complexity.

To optimize H-QAM, we first fix the constellation labels. Gray labelling may result in the best BICM capacity \cite{BICM}, and hence we try to approach Gray labelling as much as possible when pre-fixing the labels (we use the term "approach" is because we will apply quasi Gray labelling to the cross shape 32-QAM). After labelling, the constellation $\mathbf{z}$ is a function of parameter pair $(d_1, d_2)$, and in effect optimizing $\mathbf{z}$ is to optimize the parameter pair where the algorithm aforementioned is also applicable.

\section{Numerical results}
To carry out simulation, we should stipulate the $S\!N\!R$ distribution. In order to take into consideration both path-loss and shadowing, the following model \cite{channel_model} is adopted in this paper,
\begin{align}\label{Ldb}
L_{dB}(r) = 130.19+37.6\mathrm{log}_{10}r+X_{\sigma},
\end{align}
where $r$ is the distance between the user and the signal source (e.g., the BS) in kilometers and $X_{\sigma}$ is a Guassian random variable reflecting a log-normal shadowing effect which has zero means and standard deviation $\sigma=8\mathrm{dB}$. $L_{dB}$ denotes the fading factor and if the transmit signal power is $p_s$ dBm and noise power is $p_n$ dBm, then the $S\!N\!R$ of a user at distant $r$ from the signal source can be computed as
\begin{align}\label{Ldb}
S\!N\!R_{dB} = p_s-p_n-L_{dB}(r).
\end{align}
Under a reasonable assumption that users are uniformly distributed within a circular with radius $\hat{r}$ kilometers, the pdf. of a user located at distance $r$ kilometers from the center can be computed as
\begin{align}\label{r_dis}
p(r) = \frac{2r}{{\hat{r}}^2}.
\end{align}
Combining equation (\ref{Ldb}) and (\ref{r_dis}) we can get the cumulative distribution for user's $S\!N\!R$.  Assuming that $p_s=66$ dBm, $p_n=-95$ dBm and $\hat{r}=4$ kilometers, we can get the lowest $snr$ covering different percentage of users in probability as shown in table \ref{snr}. In the following, we denote an $snr$ pair as ($snr_H,snr_L$).

\begin{table}
\caption{the lowest SNR value for different percentage of user coverage.}
\vspace{-0.4cm}
\label{snr}
\begin{center}
\begin{tabular}{|c|*{7}{c|}}
\hline
 \multicolumn{1}{|c|}{\textbf{Coverage}} &\multicolumn{1}{|c|}{\textbf{95\%}} &\multicolumn{1}{|c|}{\textbf{90\%}}&\multicolumn{1}{|c|}{\textbf{85\%}}&\multicolumn{1}{|c|}{\textbf{70\%}}&\multicolumn{1}{|c|}{\textbf{60\%}}&\multicolumn{1}{|c|}{\textbf{50\%}} &\multicolumn{1}{|c|}{\textbf{40\%}} \\
\hline
\multirow{1}{*}{\textbf{SNR}(dB)}\!\!&$-0.32$&2.92&5.20&10.05&12.71&15.29&18.00\\
\hline
\end{tabular}
\end{center}
\end{table}
In the simulation, we consider two scenarios: in the first scenario, HP data covers $90\%$ users and LP data covers $70\%$ users, and in the second scenario, HP data covers $90\%$ users and LP data covers $50\%$ users. For the first scenario, we choose $m_H=m_L=2$ and for the second scenario we choose $m_H=2, m_L=3$.
\begin{figure}[t]
\centering
\includegraphics[width=3.5in]{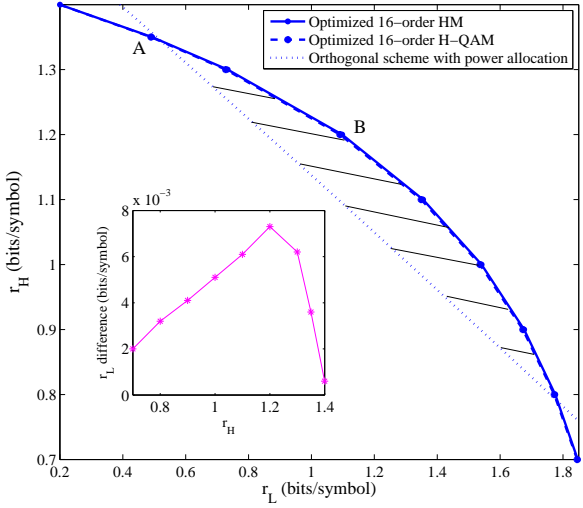}
\caption{Achievable rate regions for ($snr_H,snr_L$) = (2.92 dB,10.05 dB).}
\label{region1}
\end{figure}

\begin{figure}[t]
\centering
\includegraphics[width=3.5in]{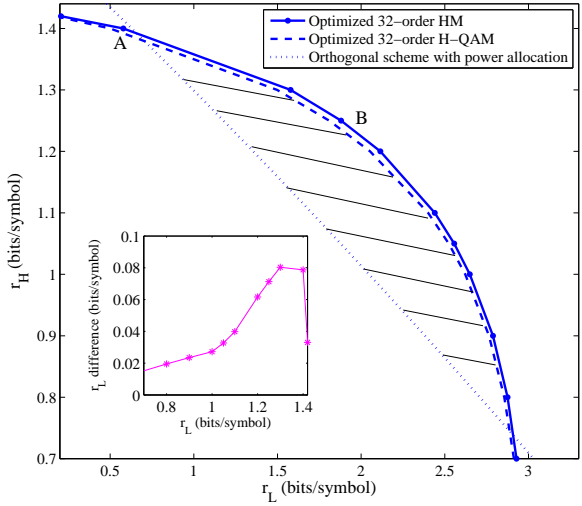}
\caption{Achievable rate regions for ($snr_H,snr_L$) = (2.92 dB, 15.29 dB).}
\label{region2}
\end{figure}

Firstly we compute the rate region, which can be achieved by our optimized hierarchical constellations, by choosing different $r^*$. As a baseline, orthogonal transmission scheme (such as TD and FD) with power allocation \cite{broadcast}\cite{wireless} is compared with hierarchical scheme, where we assume ideal Gaussian signals are adopted by orthogonal scheme.

The achievable regions are shown in Fig. \ref{region1} and Fig. \ref{region2}. It is illustrated that $r_H$ decreases as $r_L$ increases, which is easy to understand because LP layer could be viewed as a kind of interference for HP layer \cite{upgrade}. We can also find that some region can only be achieved by hierarchical modulation. In another word, the shadow region cannot be achieved by orthogonal scheme, even the signal in orthogonal scheme is ideal Gaussian. And we can use finite order HM to achieve the region which cannot be achieved by orthogonal scheme. It should be noted that if we want to get a larger achievable region, we can combine orthogonal scheme and hierarchical scheme \cite{Adaptive}, and the region of the combining scheme is the convex hull of the region achieved by orthogonal scheme and hierarchical scheme.

\begin{figure}[t]
\centering
\includegraphics[width=3.5in]{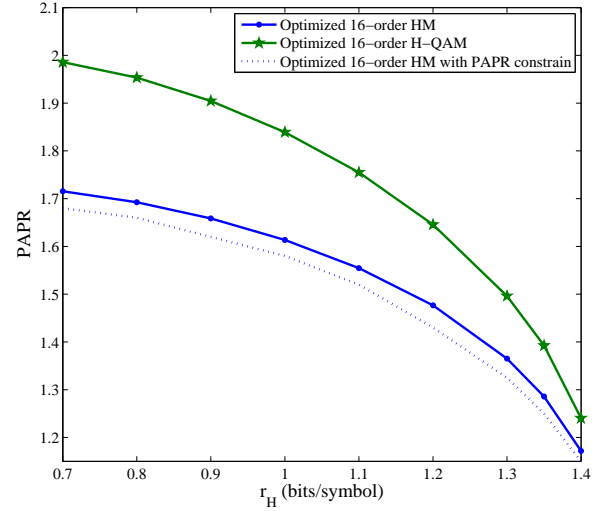}
\caption{PAPR of the two kinds of optimized 16-order hierarchical constellation for ($snr_H,snr_L$) = (2.92 dB, 10.05 dB).}
\label{PAPR1}
\end{figure}

\begin{figure}[t]
\centering
\includegraphics[width=3.7in]{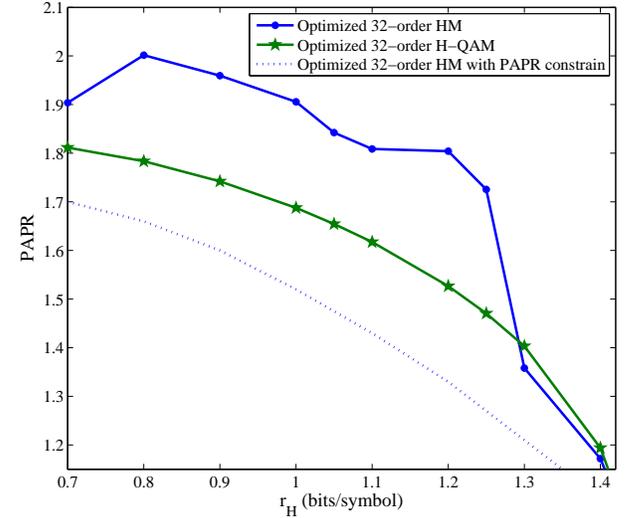}
\caption{PAPR of the two kinds of optimized 16-order hierarchical constellation for ($snr_H,snr_L$) = (2.92 dB, 15.29 dB).}
\label{PAPR2}
\end{figure}

As the magenta lines illustrate in Fig. \ref{region1} and Fig. \ref{region2}, the achievable region of optimized H-QAM are very close to the the achievable region of optimized HM with the same order. That is to say, if only taking into consideration the achievable rate region, the advantage over optimized H-QAM by using optimized HM is marginal. However, each of the two kinds of hierarchical constellations has its own advantages. The optimized H-QAM are able to be detected with lower complexity, because its in-phase and quadrature-phase parts can be viewed as two independent PAM. The optimized HM are able to possess lower PAPR, which is quite important when non-liner amplifier is adopted. The optimized 16-order HM is with lower PAPR than 16-order H-QAM is as illustrated in Fig. \ref{PAPR1}. For higher order, the PAPR of optimized HM may be larger than that of H-QAM. However we could sacrifice some BICM capacity by adding PAPR constrain to design a HM constellation with satisfied PAPR. In Fig. \ref{PAPR1} and Fig. \ref{PAPR2}, the dash lines denote the PAPR we can reach by optimized HM whose BICM capacity is equal to the BICM capacity of optimized H-QAM. And it is not hard to draw the conclusion that we can restrict the PAPR to be the same as that of optimized H-QAM, and get an optimized HM which possesses better BICM capacity. Furthermore, we can notice that when $r_H$ becomes large enough, some constellation symbols of optimized HM are superimposed and hence the PAPR of the optimized HM will be smaller than that of optimized H-QAM.
\begin{figure}
  \centering
  \subfigure[Optimized constellation for point A]{
    \label{fig:subfig:16a} 
    \includegraphics[width=2.5in]{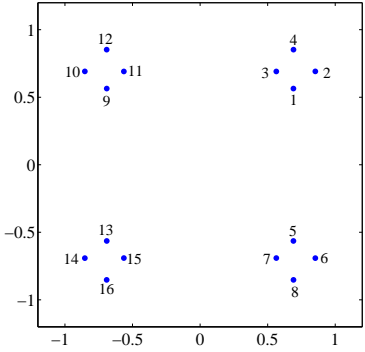}}
  \hspace{-0.5in}
  \subfigure[Optimized constellation for point B]{
    \label{fig:subfig:16b} 
    \includegraphics[width=2.4in]{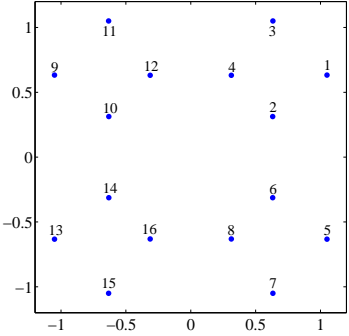}}
  \caption{Optimized 16-order hierarchical constellations for point A and B in Fig. \ref{region1}}
  \label{constellaion_10_05_2_92} 
\end{figure}

\begin{figure}
  \centering
  \subfigure[Optimized constellation for point A]{
    \label{fig:subfig:32a} 
    \includegraphics[width=2.5in]{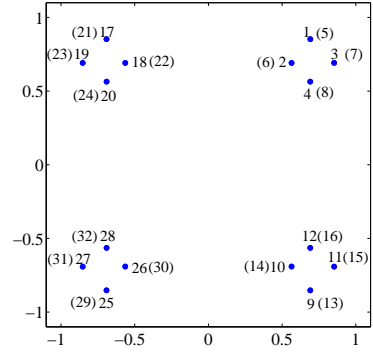}}
  \hspace{-0.5in}
  \subfigure[Optimized constellation for point B]{
    \label{fig:subfig:32b} 
    \includegraphics[width=2.5in]{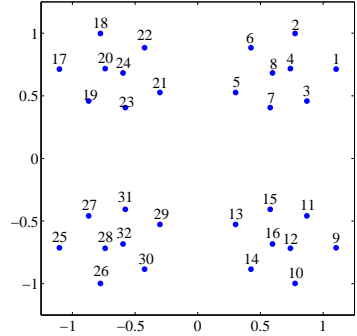}}
  \caption{Optimized 32-order hierarchical constellations for point A and B in Fig. \ref{region2}}
  \label{constellaion_10_05_2_92_2} 
\end{figure}

The optimized rate pairs can be found at the edge of achievable region, and we show some of the optimized HM which can achieve certain point of the edge. Point A in Fig. \ref{region1} and Fig. \ref{region2} is where the rate region edges of HM and orthogonal scheme approximately intersect. The point B in Fig. \ref{region1} is the point where $r_H\approx r_L$ and the point B in Fig. \ref{region2} is the point where $3r_H\approx 2r_L$. All the points are shown in Fig. \ref{constellaion_10_05_2_92} and Fig. \ref{constellaion_10_05_2_92_2}. As expected, the constellation symbols for all optimized HM are divided into 4 cluster which denotes $m_H=2$. And when $r_H$ becomes large enough, optimized 32-order HM are superimposed into 16 symbols constellation as shown in Fig. \ref{fig:subfig:32a}.

\section{conclusion}
In this paper, we address the problem of constellation optimization for HM to be used in AWGN channel. First, we propose an optimization criterion (or say an optimization goal) based on BICM-SIC capacity. Then multi-start interior point algorithm is used to approximately solve the non-convex optimization problem and some methods to reduce optimization complexity are also presented. Through numerical results, we can see that the performance of optimized HM is able to outperform that of optimized H-QAM, in terms of both rate region and PAPR.

\begin{appendix}

\subsection{Proof of Proposition \ref{ccc}}
 We first prove that simply amplifying the hierarchical constellation $\mathbf{z}$ to $\rho \mathbf{z}$ will not decrease achievable rate, where $\rho>1$ and $\rho \mathbf{z}$ means every element of $\mathbf{z}$ is multiplied by $\rho$. In order to prove that, we use an information processing model as shown in Fig. (\ref{proof}). Therein $N_1\sim\mathcal{N}_c(0,\sigma_{L}^{2})$, $N_2\sim\mathcal{N}_c(0,\sigma_{H}^{2}-\sigma_{L}^{2})$, $\hat{N}_1\sim\mathcal{N}_c(0,\sigma_{L}^{2}-\frac{\sigma_{L}^{2}}{\rho})$ and  $\hat{N}_2\sim\mathcal{N}_c(0,\sigma_{H}^{2}-\frac{\sigma_{H}^{2}}{\rho})$. $Y_H$ and $Y_L$ is the signal the receiver can get when sending constellation $\mathbf{z}$, and $\hat{Y}_H$ and $\hat{Y}_L$ is the signal the receiver can get when sending constellation $\rho \mathbf{z}$.

It is easy to know that $(\mathbf{B}_L,\mathbf{B}_H)$, $Z$, $\hat{Y}_H$ and $Y_H$ form a Markov chain $(\mathbf{B}_L,\mathbf{B}_H)\rightarrow Z\rightarrow \hat{Y}_H \rightarrow Y_H$, and hence $(\mathbf{B}_L,\mathbf{B}_H)$, $\hat{Y}_H$ and $Y_H$ also form a Markov chain $(\mathbf{B}_L,\mathbf{B}_H)\rightarrow \hat{Y}_H \rightarrow Y_H$. We can get that $B^i_H$, $\hat{Y}_H$ and $Y_H$ also form a Markov chain $B^i_H\rightarrow \hat{Y}_H \rightarrow Y_H$ which is shown as following,
\begin{align}\label{Markov}
p(b^i_H,\hat{y}_H,y_H) =p(b^i_H)p(\hat{y}_H|b^i_H)p(y_H|\hat{y}_H)\nonumber
\end{align}
The second equation in (\ref{Markov}) is based on the aforementioned fact that $(\mathbf{B}_L,\mathbf{B}_H)$, $\hat{Y}_H$ and $Y_H$ form a Markov chain. According to data-processing inequality \cite{IT}, we have following inequality,
\begin{align}
I(B^{i}_H;\hat{Y}_H)\geq I(B^{i}_H;Y_H)
\end{align}
When given $\mathbf{B}_H=\mathbf{b}_H$, we can also derive that $B^j_L$, $\hat{Y}_L$ and $Y_L$ form a Markov chain,
\begin{align}\label{Markov2}
p(b^j_L,\hat{y}_L,y_L|\mathbf{b}_H) = p(b^j_L|\mathbf{b}_H)p(\hat{y}_L|b^j_L,\mathbf{b}_H)p(y_L|\hat{y}_L)
\end{align}
\begin{figure}[t]
\centering
\includegraphics[width=3.5in]{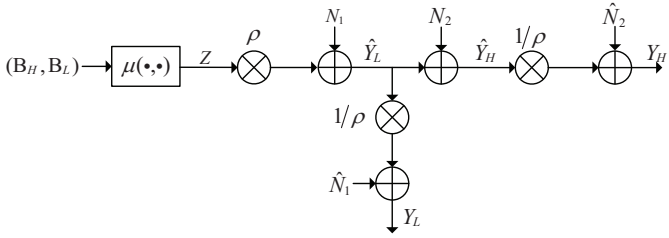}
\caption{Signal processing model}
\label{proof}
\end{figure}
The second equation in (\ref{Markov2}) is following that $(\mathbf{B}_L,\mathbf{B}_H)$, $\hat{Y}_L$ and $Y_L$ form a Markov chain $(\mathbf{B}_L,\mathbf{B}_H)\rightarrow \hat{Y}_L \rightarrow Y_L$. Based on equation (\ref{Markov2}),we have
\begin{align}
I(B^j_L;\hat{Y}_L|\mathbf{b}_H)\geq I(B^j_L;Y_L|\mathbf{b}_H)
\end{align}
For $I(\mathbf{B}^j_L;\hat{Y}_L|\mathbf{B}_H)$ is mean value of $I(\mathbf{B}^j_L;\hat{Y}_L|\mathbf{b}_H)$ with respect to $z_H$, we have
\begin{align}
I(\mathbf{B}^j_L;\hat{Y}_L|\mathbf{B}_H)&=\mathbb{E}_{\mathbf{B}_H}[I(\mathbf{B}^j_L;\hat{Y}_L|\mathbf{b}_H)]\nonumber \\
&\geq \mathbb{E}_{\mathbf{B}_H}[I(\mathbf{B}^j_L;Y_L|\mathbf{b}_H)]\nonumber \\
&= I(\mathbf{B}^j_L;Y_L|\mathbf{B}_H)
\end{align}
We can conclude here, simply amplifying a constellation will not decrease $r_H$ and $r_L$. This implies that when increasing the power of a constellation $r_H$ and $r_L$ will not decrease, as following,
\begin{align}\label{rh}
r_H(\mathbf{\hat{z}})\geq r_H(\rho \mathbf{z}) \geq r_H(\mathbf{z})
\end{align}
where $r_H(\mathbf{v})$ means the $r_H$ value when sending constellation $\mathbf{v}$, and $\hat{z}$ is the constellation that can achieve the best $r_H$ value and has the same power of constellation $\rho \mathbf{z}$. In the same way, we can get
\begin{align}\label{rl}
r_L(\mathbf{\hat{z}})\geq r_L(\rho \mathbf{z}) \geq r_L(\mathbf{z})
\end{align}

Then we assume constellation $\mathbf{z}$ satisfies $\mathbb{E}[\parallel \!Z\!\parallel^2]< p$ and $r_H\geq r^{*}$ and maximize $r_L$. Following equation (\ref{rh}) and (\ref{rl}) we can always find a constellation $\mathbf{\hat{z}}$ which satisfies $\mathbb{E}[\parallel \!\hat{Z}\!\parallel^2]= p$ and can achieve both  $r_L(\mathbf{\hat{z}}) \geq r_L(\mathbf{z})$ and constrain $r_H(\mathbf{\hat{z}})>r^{*}$.

Based on the deduction above, under both constrains $\mathbb{E}[\parallel \!Z\!\parallel^2]\leq p$ and $\mathbb{E}[\parallel \!Z\!\parallel^2]= p$, we can achieve the same maximization and proof is completed.

\end{appendix}
\bibliographystyle{IEEEtran}
\bibliography{paper}

\end{document}